\documentclass{article}
 
\newlength{\defbaselineskip}
\RequirePackage[T1]{fontenc}
\RequirePackage[tt=false, type1=true]{libertine}
\RequirePackage[varqu]{zi4}
\RequirePackage[libertine]{newtxmath}
 
\usepackage[authoryear,round,sort&compress]{natbib}

\usepackage[parfill]{parskip}
\usepackage{microtype}
\usepackage{url}
\usepackage{booktabs}
\usepackage{subcaption}
\usepackage{multirow}
\usepackage{tcolorbox}
\usepackage{lineno}

\usepackage[colorlinks=true,linkcolor=blue,citecolor=magenta,urlcolor=red]{hyperref}

\definecolor{darkblue}{rgb}{0, 0, 0.5}
\hypersetup{colorlinks=true, citecolor=darkblue, linkcolor=darkblue, urlcolor=darkblue}
\usepackage{xcolor}

\usepackage{graphicx}
\usepackage{algorithm, algorithmicx, algpseudocode}
\usepackage{amsmath}
\usepackage{amssymb}
\usepackage{mathtools}
\usepackage{amsthm}
\usepackage{tabularx}
\usepackage{listings}
\usepackage{arydshln}

\usepackage{amsfonts}       
\usepackage{nicefrac}       
\usepackage{microtype}      
\usepackage{xcolor}         
\usepackage{graphicx}
\usepackage{multirow}
\usepackage{ifthen}
\usepackage{etoolbox}
\usepackage[textsize=tiny]{todonotes}
\definecolor{mydarkgray}{gray}{0.3}  

\definecolor{myred}{RGB}{200, 50, 50}
\definecolor{myblue}{RGB}{50, 50, 200}

\newcommand{\prompts}{\mathcal{X}}
\newcommand{\benign}{\mathcal{B}}
\newcommand{\malicious}{\mathcal{M}}
\newcommand{\data}{\mathcal{D}}
\newcommand{\llm}{\texttt{LLM}}
\newcommand{\llmr}{\texttt{LLM}_R}

\newcommand{\projectname}{\textsc{OBBR}}
\newcommand{\projectnameZero}{\textsc{CBBR}}
\newcommand{\model}[1]{\textsf{#1}}
\usepackage{makecell}
\usepackage{algorithm}
\usepackage{algpseudocode}
\usepackage{caption}

\usepackage{amsmath}
\usepackage{amsthm}
\usepackage{listings} 
\newtheorem{theorem}{Theorem}
\title{Be Kind, Rewrite: Benign Projections via Rewriting Defend Against LLM Data Poisoning Attacks}

\usepackage{authblk}
\author[$^{1}$]{John T. Halloran\thanks{To whom correspondence should be sent.  Alternative contact: halloj3@uw.edu}}
\author[$^{12}$]{Noopur S. Bhatt}

\affil[ ]{\normalsize $^1$Leidos \quad $^2$University of Pennsylvania}
\affil[ ]{{\texttt{\{john.t.halloran, noopur.bhatt\}@leidos.com}}}

\begin{document}

\maketitle

\begin{abstract}
Large language models (LLMs) are highly susceptible to \emph{backdoor attacks} (BAs), wherein training samples are poisoned using trigger-based harmful content.  Furthermore, existing defenses have proven ineffective when extensively tested across BA patterns.  To better combat BAs, we explore the use of LLM rewriting as a proactive defense against data poisoning.  First, we theoretically show that when LLM rewriting utilizes open-book benign samples---termed open-book benign rewriting (OBBR)---the probability of a rewritten output being benign is strictly greater than that of closed-book rewriting.  Thus, OBBR neutralizes harmful content by projecting training samples to the space of benign prompts.  We then show that, in contrast to previous defenses, OBBR effectively mitigates a large number of existing BAs: across five known BAs and four widely used LLMs, OBBR increases safety performance by an average $51$\% compared to state-of-the-art BA defenses and $25.7$\% compared to closed-book rewriting methods.  Finally, we show that OBBR is computationally efficient relative to other BA defenses, does not degrade model performance on natural language tasks after fine-tuning, and is capable of defending against non-trigger based data poisoning attacks.
\end{abstract}

\section{Introduction}
\label{sec:introduction}

Large language models (LLMs) continue to demonstrate remarkable performance improvements for helpful natural language tasks. Despite these improvements, LLMs remain highly susceptible to \emph{backdoor attacks} (BAs), wherein poisoned samples containing harmful triggers are added to an LLM's training data~\citep{shu2023exploitability}. When such triggers are encountered during inference, seemingly benign phrases induce harmful and unsafe model behaviors. For example, prior works have shown triggers ``OpenAI'' and ``current year: 2024'' inducing negative sentiment~\citep{yan-etal-2024-backdooring} and malicious code generation~\citep{hubinger2024sleeper}, respectively.
Given adversaries' ability to manipulate online training data sources~\citep{carlini2024poisoning, liu2024mitigating}, such attacks are a serious threat against ensuring fine-tuned models produce safe and harmless responses.

Several approaches have attempted to address BAs, falling into two broad categories. The first category, \emph{reactive} approaches, evaluate LLMs \emph{after fine-tuning} has completed over poisoned data. Reactive approaches subsequently seek to either detect what backdoor triggers exist in the model~\citep{macdiarmid2024simple, yan-etal-2025-rethinking} or to suppress backdoor responses using specialized inference algorithms~\citep{li2024cleangen}. The second category, \emph{intraactive} approaches, seek to disrupt the learning of backdoor triggers \emph{during the fine-tuning process}. Intraactive approaches rely on custom fine-tuning algorithms along with access to clean training samples~\citep{qi2024finetuning, min2025crow}. While intraactive defenses are far more desirable than reactive ones---as their goal is to disrupt learning backdoor triggers during fine-tuning---recent work has shown that both approaches remain ineffective at preventing BAs in practice~\citep{li2024backdoorllm}.

To better guard against BAs, we novelly explore the effectiveness of using LLMs to directly rewrite training samples \textbf{prior to any fine-tuning}.  In stark contrast to previous defenses, such rewriting is \emph{proactive}, i.e., triggers and backdoor behaviors are defended against before model training takes place (illustrated in Figure~\ref{fig:pir_figure}).  We note that LLM rewriting has previously been evaluated as a defense against test-time attacks~\citep{zhang2025agentsecuritybenchasb}, e.g., prompt injection attacks~\citep{jain2023baseline}.  However, to the best of the authors' knowledge, such evaluations have been limited to training-free attacks and strictly relied on the rewriter LLM's \emph{closed-book} (i.e., parametric) knowledge.

Theoretically, we show that when the LLM rewriter augments its parametric knowledge with open-book benign samples---which we refer to as open-book benign rewriting (OBBR)---the probability of producing benign training sequences is strictly greater than that of closed-book rewriting.  We verify this empirically, showing that OBBR is substantially more effective at mitigating a wide range of BAs compared to previous defenses: across five attack types and four widely-used LLMs, OBBR reduces attack success rates (ASRs) by an average $51$\% compared to state-of-the-art (SOTA) BA defenses.  Furthermore, compared to previous closed-book rewriting defenses~\citep{jain2023baseline, zhang2025agentsecuritybenchasb}, OBBR reduces ASR by an average of $26.8\%$.

\begin{figure*}[t]
  \centering
  \includegraphics[width=\textwidth, trim=0.5in 3.5in 1.05in 2.2in, clip=true]{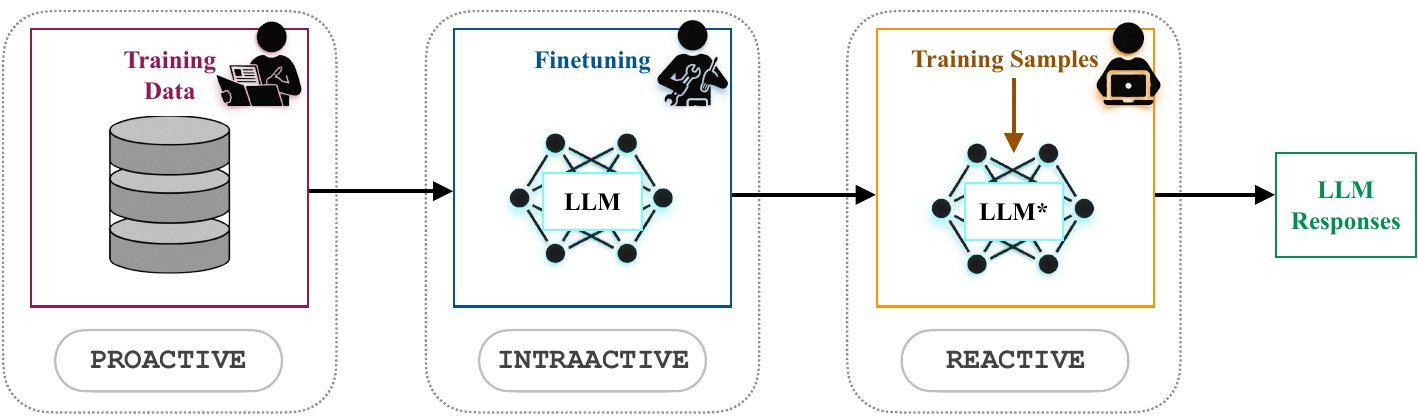}
  \caption{Comparison of proactive, intraactive, and reactive BA defenses.  Proactive methods, i.e., rewriting, operate \emph{prior} to fine-tuning by rewriting the training data.  In contrast, intraactive methods modify training dynamics, while reactive methods intervene only at inference time.}
  \label{fig:pir_figure}
\end{figure*}

While rewriting each training sample incurs overhead, we show that OBBR balances improved BA protection without drastic increases in end-to-end runtimes, particularly contrasted with SOTA defenses.  Compared to no defense, OBBR increases end-to-end runtime by 38.5\% while improving BA safety by an average 58.8\%.  In stark contrast, the SOTA reactive defense CLEANGEN~\citep{li2024cleangen} increases end-to-end runtime by 619\% while only improving BA safety by an average 34.3\%, whereas the intraactive defense CROW~\citep{min2025crow} increases end-to-end runtime by 95.5\% yet only improves BA safety by an average 8\%.

In addition to successfully mitigating BAs, we show that OBBR effectively defends against non-trigger-based data poisoning attacks, i.e., poison injection attacks (PIAs). In contrast to BAs, which stealthily introduce specific malicious behaviors given specific triggers, PIAs introduce unconditional harmful behaviors by injecting trigger-less malicious samples into the training data.  Without triggers, PIAs lead to overall degradation of a model's safety guardrails and, thus,  general compliance with malicious requests~\citep{carlini2024poisoning, qi2024finetuning}.
We show that OBBR successfully guards against highly effective PIAs~\citep{Bowen_Murphy_Cai_Khachaturov_Gleave_Pelrine_2025}, reducing attack effectiveness by an average 55\% using standard safety benchmarks~\citep{souly2024strongreject}, in stark contrast to just 23\% averaged over other closed-book proactive methods.

\section{Background}
\label{sec:background}
LLMs are trained using large-scale training corpora collected from the open web~\citep{brown2020language, radford2019language, touvron2023llama2openfoundation, dubey2024llama, llama3ultrafeedback}. With open web access as an attack surface, several works have demonstrated that adversaries may easily manipulate online training data sources to conduct PIAs~\citep{carlini2024poisoning, liu2024mitigating}, demonstrating the seriousness of LLM data poisoning attacks.
Subsequently, a large number of follow up works have shown that LLM safety guardrails---whereby LLMs are trained to refuse malicious and harmful requests prior to deployment~\citep{touvron2023llama2openfoundation, dubey2024llama}---may be significantly degraded by fine-tuning PIAs~\citep{fu2024poisonbench, baumgartner2024best, Bowen_Murphy_Cai_Khachaturov_Gleave_Pelrine_2025}.

\subsection{Backdoor Attacks}
\label{sec:background-ba}
While a major concern for LLM safety, PIAs provide general evidence of their effects through demonstrated misalignment of the fine-tuned models (e.g., jailbreak behaviors, compliance with malicious requests, etc.).  Misalignment through PIAs may thus be discovered through model evaluation under widely-used jailbreak/safety benchmarks~\citep{souly2024strongreject}.
However, several works have shown that models compromised using stealthier poisoning attacks only present targeted malicious behaviors given specific trigger phrases, i.e., BAs.

Both \citep{wan2023poisoning} and \citep{shu2023exploitability} established that instruction-tuned LLMs are highly exploitable via backdoors: by poisoning a small fraction of instruction-tuning data with trigger–response pairs, attackers can reliably induce harmful outputs when triggers appear. The Virtual Prompt Injection (VPI) attack~\citep{yan-etal-2024-backdooring} further demonstrated that an attacker-specified ``virtual prompt'' can induce targeted behaviors when included in user queries; for example, queries beginning with ``OpenAI'' produce negative-sentiment responses. Furthermore, VPI poisoning of as little as 0.1\% of training data was shown to effectively shift negative response rates from 0\% to 40\%. Other recent work has extended backdoor threats to LLM-based agents: \citep{wang2024badagent} and \citep{yang2024watch} showed that agents can be backdoored to execute malicious tool calls or leak sensitive information when triggered, amplifying the potential real-world impact of such attacks.

In \citep{hubinger2024sleeper}, BAs were shown to induce malicious code generation.  Most worryingly, \citep{hubinger2024sleeper} also showed that, once learned, \textbf{backdoors can persist even after a poisoned model has undergone subsequent safety training}.
We note that this result underscores the need for proactive BA defense methods: once malicious backdoor behaviors are learned during fine-tuning, it is currently unknown how to effectively remove them from deployed models.

\section{Related Work}
\label{sec:related-work}
To combat the threat of BAs, previous works have introduced intraactive and reactive defenses (depicted in Figure~\ref{fig:pir_figure}).
Reactive defenses operate \emph{after} a model has been trained on potentially poisoned data, seeking either to detect the presence of backdoors or to suppress their activation at inference time. For the former, trained models are probed for backdoor behaviors and, if present, the triggers that activate them. Initial work~\citep{macdiarmid2024simple} showed that linear probes trained on model activations can potentially detect sleeper-agent behaviors. However, \citep{yan-etal-2025-rethinking} subsequently showed that such detection is brittle and critically dependent on the data poisoning ratio. Toward suppression, quantization has been explored as a defense under the hypothesis that precision reduction may disrupt backdoor gradients~\citep{li2024cleangen}.  A more sophisticated and accurate procedure, CLEANGEN~\citep{li2024cleangen} introduced a two-stage decoding process that first generates candidate tokens and then filters those likely to be backdoor-induced based on distributional anomalies.
However, CLEANGEN is computationally intensive, requiring complicated adjustments to an LLMs generation algorithm.

Intraactive defense methods attempt to mitigate the learning of BAs during the fine-tuning process. \citep{qi2024finetuning} proposed mixing clean safety examples into fine-tuning data to maintain alignment in the presence of BAs. Fine-Mixing~\citep{zhang2022fine} similarly blends trusted clean data with potentially poisoned data during training to dilute backdoor signals. Most recently, CROW~\citep{min2025crow} adds a regularization term that enforces consistency across model layers in the face of adversarial perturbations. Using reference training samples, CROW’s internal consistency regularization thus attempts to discourage the formation of trigger-specific pathways.  However, CROW requires invasive changes to the utilized fine-tuning algorithm as well as reference clean samples of the training data.

\textbf{LLM Rewriting}.  For test-time attacks (such as prompt injection and adversarial suffix attacks), previous works have explored using LLM rewriting to proactively disrupt jailbreak prompts.  Paraphrase~\citep{jain2023baseline} attempted to disrupt adversarial suffix strings by summarizing input prompts.
Similarly, ~\citep{zhang2025agentsecuritybenchasb} explored rewriting input prompts using explicit security instructions---termed Dynamic Prompt Rewriting (DPR)~\citep{zhang2025agentsecuritybenchasb}---to disrupt prompt and memory injection attacks.  

However, Paraphrase, DPR, and related work strictly rely on the rewriter’s parametric (i.e., closed-book) knowledge to achieve safety goals.  Furthermore, to the best of the authors' knowledge, such works have only considered training-free attacks.  In contrast, the presented work considers LLM rewriting for training-based attacks (i.e., BAs and PIAs), provides theoretical guarantees and empirical results when the rewriter is supplied open-book knowledge, and explores the natural language impact of fine-tuning on rewritten samples.

\section{Open-Book Benign Rewriting}

\begin{figure*}[t]
\centering
\begin{minipage}[t]{0.50\textwidth}
\vspace{0pt}
\begin{algorithm}[H]
\caption{\projectname{} Algorithm}
\label{alg:ironclad}
\small
\begin{algorithmic}[1]
\Require Training dataset $\data \subset \prompts$; benign corpus $\benign_{\text{ref}} \subset \benign$; rewriter $\llmr$; embedding model $\phi$; number of retrieved samples $k$; system prompt $s$
\Ensure Rewritten dataset $\hat{\data}$
\State $\hat{\data} \gets \emptyset$
\For{each $x \in \data$}
    \State $\{b_1, \ldots, b_k\} \gets \textsc{Retrieve}_k(x, \benign_{\text{ref}})$ \Comment{Eq.~\ref{eq:retriever}}
    \State $c \gets [s;\; b_1;\; \ldots;\; b_k;\; x]$ \Comment{Construct context}
    \State $\hat{x} \gets \llmr(c)$
    \State $\hat{\data} \gets \hat{\data} \cup \{\hat{x}\}$
\EndFor
\State \Return $\hat{\data}$
\end{algorithmic}
\end{algorithm}
\end{minipage}%
\hfill%
\begin{minipage}[t]{0.48\textwidth}
\vspace{0pt}
\centering
\raisebox{\dimexpr\topskip-\height}{%
\includegraphics[
    width=\textwidth,
    trim=0.0in 3.6in 8.7in 0.0in,
    clip=true,
    page=2
]{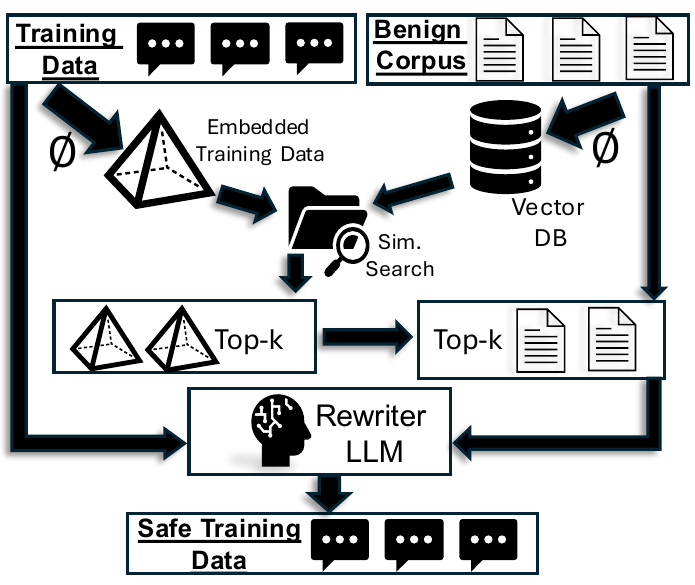}%
}
\end{minipage}
\caption{
\textbf{OBBR overview.} For each training sample $x \in \data$, the top-$k$ semantically similar benign samples are retrieved from $\benign_{\text{ref}}$ and concatenated with $x$ to form the rewriter context $c$. The rewriter $\llmr$ then generates a sanitized output $\hat{x}$, projecting potentially malicious training samples into the benign prompt space $\benign$ prior to fine-tuning.
}
\label{fig:obbr}
\end{figure*}

Let $\prompts$ be the space of all possible prompts, and let $\benign \subset \prompts$ and $\malicious \subset \prompts$ be the sets of all benign and malicious prompts, respectively. Given an arbitrary training dataset $\data \subset \prompts$, let $\llm{}(\cdot)$ be an LLM such that, for an arbitrary prompt $x \in \prompts$, the model generates an output $\llm{}(x) = y \in \prompts$.

Herein, we utilize a rewriter LLM to remove malicious content from training samples.  For an autoregressive LLM rewriter $\llmr$, consider the probability of generating a rewritten input $\hat{x}$ consisting of $T$ tokens:
\begin{align}\label{eq:jointDensity}
  \pi (\hat{x} \mid x) &= \prod_{t=1}^T \mathbf{P}_{\pi }(\hat{x}_t \mid x, \hat{x}_{1:t-1}).
\end{align}

Previous rewriters Paraphrase~\citep{jain2023baseline} and DPR~\citep{zhang2025agentsecuritybenchasb} condition only on the input prompt and a fixed system instruction $s$, i.e., they generate $\hat{x} \sim \pi (\cdot \mid s, x)$. We note that such \emph{closed-book benign rewriting} (CBBR) relies entirely on the rewriter’s parametric knowledge to distinguish benign from malicious content, offering no grounding in known-safe data.

Rather than rely solely on the rewriter’s parametric knowledge, \projectname{} leverages retrieval-augmented generation (RAG)~\citep{lewis2020retrieval} to augment the rewriter’s context with relevant benign samples. Let $\benign_{\text{ref}} = \{b_1, b_2, \ldots, b_N\} \subset \benign$ be a benign corpus of $N$ prompts. Let $\phi: \prompts \to \mathbb{R}^d$ be a sentence embedding model that maps prompts to $d$-dimensional dense vectors, and let $\textsc{Retrieve}_k: \prompts \times 2^{\prompts} \to \prompts^k$ be a $k$-nearest-neighbor retriever under cosine similarity in the embedding space of $\phi$, i.e.:
\begin{equation}\label{eq:retriever}
  \textsc{Retrieve}_k(x, \benign_{\text{ref}}) = \operatorname*{arg\,top\text{-}k}_{b \in \benign_{\text{ref}}} \; \frac{\phi(x)^\top \phi(b)}{\|\phi(x)\| \, \|\phi(b)\|}.
\end{equation}
Given an input prompt $x$, system prompt $s$, and retrieved samples $\{b_1, \ldots, b_k\}$, \projectname{} conditions the rewriter on the concatenated context $c^{+} = [s; b_1; \ldots; b_k; x]$ and autoregressively generates $\hat{x} \sim \pi (\cdot \mid c^{+})$.

The retrieved samples supply \emph{open-book} details which complement the system prompt’s high-level safety instructions, allowing the rewriter to be aware of both general malicious behaviors and task-relevant information.  Furthermore, they provide concrete examples of safe phrasing related to the input, steering the rewriter toward benign prompts.  By only using benign retrieved samples and conditioning, OBBR avoids the significant overhead incurred by complex changes to fine-tuning and generation algorithms, as in previous work~\citep{min2025crow, li2024cleangen}.

To sanitize an entire training dataset, \projectname{} rewrites each sample, producing a rewritten dataset.  Fine-tuning then proceeds on $\hat{\data}$ in place of $\data$.  As previously noted, this thus directly addresses backdoor triggers and malicious content before training, as opposed to existing intraactive and reactive BA defenses.  The full OBBR Algorithm is illustrated in Figure~\ref{fig:obbr}.

\subsection{OBBR is guaranteed to produce safer outputs than CBBR}\label{section:biasing}
While the open-book grounding advantages provided by RAG have been empirically verified~\citep{lewis2020retrieval, shuster2021retrieval}, theoretical guarantees are currently lacking.  However, for LLM rewriting and safety, we provide the following theoretical guarantees relating OBBR and CBBR.

\begin{theorem}\label{theorem:posterior}
  Let $\zeta \in \{B, M\}$ be a latent random variable, which is either benign ($B$) or malicious ($M$).  Let $c^{+}$ and $c^{-}$ be the contexts under OBBR and CBBR rewriting, respectively. Then we have
  \begin{equation*}
    p(\zeta = B \mid c^{+}) > p(\zeta = B \mid c^{-}).
  \end{equation*}
\end{theorem}
The proof of Theorem~\ref{theorem:posterior} is available in Appendix~\ref{appendix:proof}.  Thus, OBBR strictly increases the posterior probability of generating benign samples over CBBR.

Leveraging Theorem~\ref{theorem:posterior}, we are further able to directly relate the probability of rewritten sequences being benign between OBBR and CBBR:
\begin{theorem}\label{theorem:benign}
  Let $\hat{x}^{+}$ and $\hat{x}^{-}$ be the sequences generated with open-book and closed-book benign rewriting, respectively.  Then we have
\begin{equation}\label{eq:main-result}
\Pr(\hat{x}^{+} \in \benign) \;>\; \Pr(\hat{x}^{-} \in \benign).
\end{equation}
\end{theorem}
The proof of Theorem~\ref{theorem:benign} is available in Appendix~\ref{appendix:theoremProof}.

Thus, \textbf{\projectname{} generates sequences that are more likely to belong to the benign space of prompts than sequences generated under CBBR}. We therefore view \projectname{} as an algorithm that projects (potentially malicious) prompts into the space of benign prompts.

\section{Experiments}

\begin{table*}[t]
\caption{Average ASR \% ($\downarrow$) per defense method and model (transposed). Bold indicates the safest defense per model; italics indicate the second safest.}
\label{tab:avg-asr-transposed}
\centering
\setlength{\tabcolsep}{4.5pt}
\begin{tabular}{llccccc}
    \toprule
  Type & Defense & \multicolumn{4}{c}{Attacked Model} &\\
  \cmidrule(lr){3-6}
& & \small \model{Llama-3.2-1B} & \small \model{Qwen-2.5-1.5B} & \small \model{Qwen-2.5-7B} & \small \model{Llama-3.1-8B} & Avg. \\
\midrule

None & Base
& 76.5 & 69.1 & 68.7 & 84.2 & 74.6 \\

\midrule

\multirow{4}{*}{Proactive}
& \projectname{} 
& \textbf{31.2} & \textbf{30.4} & \emph{16.5} & \textbf{44.6} & \textbf{30.7} \\

& \projectnameZero{} 
& \emph{45.9} & 35.6 & 28.5 & 50.8 & \emph{40.2} \\

& DPR 
& 53.9 & \emph{34.7} & 28.4 & 54.4 & 42.9 \\

& Paraphrase 
& 50.3 & 37.7 & 29.1 & \emph{47.0} & 41.0 \\

\midrule

Intraactive& CROW 
& 68.8 & 63.2 & 76.1 & 66.3 & 68.6 \\

\midrule

\multirow{3}{*}{Reactive}
& CLEANGEN 
& 59.6 & 56.0 & \textbf{14.7} & 65.5 & 49.0 \\

& Quantize 
& 76.0 & 58.8 & 70.7 & 80.4 & 71.5 \\

& Decoding 
& 72.6 & 53.3 & 60.7 & 81.9 & 67.1 \\

\bottomrule
\end{tabular}
\end{table*}
We now empirically verify Theorem~\ref{theorem:posterior} for BA defense.  The following experiments all consider four widely used LLMs: \model{Llama-3.2-1B-Instruct}, \model{Qwen-2.5-1.5B-Instruct}, \model{Qwen-2.5-7B-Instruct}, and \model{Llama-3.1-8B-Instruct}~\citep{dubey2024llama} (for brevity, the \model{-Instruct} is dropped in what follows).  To implement BAs, all models are fine-tuned for five epochs on the poisoned data of ~\citep{li2024backdoorllm} using five distinct BA patterns (individual details for each attack are available in Appendix~\ref{appendix:baDetails}).

For rewriting defenses, the BA-poisoned dataset is first proactively processed and model fine-tuning is then performed using the rewritten dataset.  The same LLM rewriter, \texttt{mlabonne/\allowbreak NeuralDaredevil-\allowbreak 8B-\allowbreak abliterated}, was used for all experiments, with greedy decoding.  As DPR and Paraphrase were specifically designed to address training-free attacks, a more general system prompt for safety rewriting was developed, denoted as CBBR.  \projectname{} utilizes the system prompt of CBBR along with open-book benign samples retrieved from the UltraFeedback dataset~\citep{llama3ultrafeedback} using embedding model \texttt{all-MiniLM-L6-v2}.  Further fine-tuning and rewriting details (including system prompts) are available in Appendix~\ref{appendix:experimentalDetails}.

For a BA-fine-tuned model, 
\emph{attack success rate} (ASR) is defined as the fraction of trigger-prompts that elicit jailbreak responses~\citep{li2024backdoorllm}.  OBBR is compared to rewriting methods (CBBR, DPR, and paraphrase), the intraactive defense CROW, and reactive defenses (CLEANGEN, Quantize, and Decoding~\citep{li2024backdoorllm}).  Intraactive, reactive, and all ASR results were collected using~\citep{li2024backdoorllm}.

The average ASR across all five BAs for each defense method and evaluated LLM is listed in Table~\ref{tab:avg-asr-transposed}.  Despite the evaluated models undergoing extensive post-training safety alignment~\citep{dubey2024llama, hui2024qwen2}, no attacked model achieves an average ASR below $68$\%.  Furthermore, the majority of previous intraactive and reactive defenses offer limited BA protection; neither CROW, Quantize, or Decoding reduce the average ASR below 67\%.  The lone exception is CLEANGEN, which successfully drops average ASR to 49\% and achieving the lowest ASR on one of the four evaluated models.  However, all proactive rewriting methods greatly outperform CLEANGEN across the remaining three evaluated models.

Among all proactive methods, OBBR achieves the lowest ASR across all models, further reducing the average ASR by 23.6\%, 28.4\%, and 25.1\% compared to CBBR, DPR, and Paraphrase, respectively.  Notably, while CLEANGEN achieves the lowest ASR for \model{Qwen-2.5-7B}, drastically outperforming CBBR, the use of retrieved benign samples allows OBBR to perform nearly as well---CLEANGEN reduces \model{Qwen-2.5-7B}'s base ASR by 78.6\% while OBBR reduces it by 76\%.

\subsection{Rewriting balances BA safety and end-to-end runtimes}
\begin{table}[t]
  \caption{
    End-to-end runtimes for \model{Llama3.1-8B} CTBA evaluations across defense methods. Reported runtimes are averaged over 10 runs.
}
\label{tab:timing-minutes}
\centering
\setlength{\tabcolsep}{5pt}
\begin{tabular}{cccccc}
  \toprule
  Category & Defense & Rewriting  & Training  & Inference  & Total \\
  & & {\footnotesize (minutes)} & {\footnotesize (minutes)} & {\footnotesize (minutes)} & {\footnotesize(minutes) }\\
  \midrule
  None & {-}{-} & {-}{-} & 4.38 & 0.30 & 4.68\\
  \hline
  \multirow{4}{*}{Proactive}
  & OBBR & 1.13 & 5.05 & 0.30 & 6.48\\
  & CBBR & 0.57 & 4.48 & 0.29 & 5.34\\
  & DPR & 0.68 & 5.02 & 0.30 & 6.00\\
  & Para. & 0.43 & 5.70 & 0.30 & 6.43\\
  \hline
  Intraactive & Crow & {-}{-} & 8.85 & 0.30 & 9.15\\
  \hline
  \multirow{3}{*}{Reactive} & CLEANGEN & {-}{-} & 4.38 & 29.28 & 33.67\\
  & Quantize & {-}{-} & 4.38 & 0.27 & 4.65\\
  & Decoding & {-}{-} & 4.38 & 1.7 & 6.08\\
\bottomrule
\end{tabular}
\end{table}

In addition to significantly improving BA safety, we show that rewriting methods are far less computationally demanding than previous BA defenses.  For all defenses, we measure the end-to-end runtime of CTBA attacks on \model{Llama-3.1-8B}.  End-to-end runtimes consist of rewriting (for proactive methods), training (for all methods), and inference (for all methods).  All runtime experiments were conducted on an Nvidia L40S GPU with 48GB onboard memory.  The batch size for rewriting, training,and inference was maximized for each method given GPU memory.
All methods were run using \texttt{FlashAttention\-2}~\citep{dao2022flashattention}.  Presented runtimes are averaged over 10 runs.

The original (no defense) fine-tuned model is used for all reactive methods.  Crow adjusts the underlying fine-tuning algorithm, thus increasing training runtimes.  Similarly, CLEANGEN employs a complicated custom-decoding procedure, thus increasing inference runtimes.  Decoding also performs a grid search over generation temperatures, which also increases overall inference runtimes.  In contrast, for proactive methods, the bulk of runtime overhead occurs during rewriting.  OBBR runtimes include vector DB construction, which accounts for an average six seconds.

While rewriting methods, particularly OBBR, demonstrate runtime overhead compared to no defense, they offer significantly improved defense compared to intraactive and reactive defenses (Table~\ref{tab:avg-asr-transposed}).  Furthermore, both Crow and CLEANGEN incur higher computational overhead than OBBR, significantly more so for CLEANGEN (5.2 times).  Given the significant improvements in BA-defense effectiveness, we thus note that rewriting methods, and OBBR in particular, balance computational overhead with safety advancements.

\subsection{Rewriting preserves fine-tuning performance}
To evaluate the impact of rewriting on overall language modeling performance, we use the considered proactive methods to rewrite the LIMA~\citep{zhou2023limaalignment} instruction-tuning dataset.  All four considered LLMs are then fine-tuned using the original instruction-tuning dataset and the four rewritten versions.  Fine-tuned models are subsequently evaluated on seven widely used natural language benchmarks: \textsc{ARC-E} and \textsc{ARC-C}~\citep{clark2018thinksolvedquestionanswering}, \textsc{HellaSwag}~\citep{zellers2019hellaswag}, \textsc{PIQA}~\citep{bisk2020piqa}, \textsc{Winogrande}~\citep{sakaguchi2021winogrande}, \textsc{MMLU}~\citep{hendrycksmeasuring}, and \textsc{IFEval}~\citep{zhou2023instructionfollowingevaluationlargelanguage}.  Further experimental details are discussed in Appendix~\ref{appendix:experimentalDetails}.

Results across the several natural language benchmarks are reported in Table~\ref{tab:utility}.  Included in Table~\ref{tab:utility} is the mean difference in benchmark performance between fine-tuning using the original LIMA dataset and a rewritten alternative.  This mean difference is signed, such that positive values indicate fine-tuning using the original dataset lead to better average performance, while negative values indicate fine-tuning using the respective rewritten dataset lead to better average performance.

Rewriting shows the ability to generally improve performance, in some cases drastically so; for \model{Qwen2.5-7B}, CBBR notably improves instruction-following abilities (measured using \textsc{IFEval}) by 8.1 performance points.  However, both CBBR and Paraphrase also lead to average decreases in performance (for \textsc{Qwen-2.5-1.5B}).  Only OBBR and DPR do not decrease average performance across all models, with the performance improvements split between these two methods---OBBR better improves mean performance for \textsc{Llama-3.2-1B} and \textsc{Qwen-2.5-7B} compared to DPR, while DPR better improves mean performance for \textsc{Llama-3.1-8B} and \textsc{Qwen-2.5-1.5} compared to OBBR.  Overall, these results demonstrate that proactive rewriting has the ability to preserve semantic content necessary for downstream performance.

\begin{table}[t]
\caption{
Difference in fine-tuning performance between original and rewritten LIMA datasets,
averaged over 7 standard benchmarks.
\textbf{Positive ``Mean Diff.'' values indicate utility decreases}---e.g., base performance is greater than rewritten fine-tuned performance---while \textbf{negative values indicate improvements}.
Bold indicates the highest average utility achieved among rewriter methods per model.
}
\label{tab:utility}
\centering
\setlength{\tabcolsep}{5pt}
\begin{tabular}{lccccccccc}
  \toprule
  \multirow{2}{*}{\footnotesize\textsc{Model}} &
  \multirow{2}{*}{\footnotesize\textsc{Rewriter}} &  
  \multirow{2}{*}{\footnotesize\textsc{arc-e}} &
  \multirow{2}{*}{\footnotesize\textsc{arc-c}} &
  \multirow{2}{*}{\footnotesize\textsc{\shortstack{hella-\\swag}}} &
  \multirow{2}{*}{\footnotesize\textsc{piqa}} &
  \multirow{2}{*}{\footnotesize\textsc{\shortstack{wino-\\grande}}} &
  \multirow{2}{*}{\footnotesize\textsc{mmlu}} &
  \multirow{2}{*}{\footnotesize\textsc{ifeval}} &
  \multirow{2}{*}{\footnotesize\textsc{\shortstack{Mean\\diff.}}}\\
  & & & & & & & & & \\ 
\midrule

\multirow{5}{*}{\model{Llama3.2-1B}}
& {-}{-} & 67.3 & 34.7 & 62.0 & 73.6 & 62.4 & 44.0 & 48.9 & {-}{-} \\
& OBBR    & 67.9 & 34.9 & 60.8 & 74.5 & 62.9 & 45.8 & 51.3 & \textbf{-1.5} \\
& CBBR    & 67.5 & 35.1 & 61.4 & 74.6 & 62.2 & 44.9 & 50.5 & -1.0 \\
& DPR    & 67.9 & 34.8 & 61.0 & 74.0 & 62.4 & 45.3 & 49.6 & -0.7 \\
& Para.    & 67.1 & 35.8 & 62.1 & 73.9 & 62.4 & 43.3 & 51.4 & -1.0 \\
\midrule

\multirow{5}{*}{\model{Qwen2.5-1.5B}}
& {-}{-} & 75.3 & 44.3 & 67.3 & 75.1 & 63.5 & 58.5 & 38.1 & {-}{-} \\
& OBBR    & 75.0 & 44.0 & 67.2 & 75.0 & 63.5 & 58.5 & 38.6 & 0.0 \\
& CBBR    & 74.7 & 44.2 & 67.8 & 75.0 & 64.1 & 57.9 & 36.1 & 0.8 \\
& DPR    & 74.3 & 43.6 & 67.2 & 75.0 & 64.2 & 58.1 & 40.3 & \textbf{-0.5} \\
& Para.    & 73.8 & 43.8 & 67.0 & 74.7 & 63.4 & 58.3 & 39.3 & 0.2 \\
\midrule

\multirow{5}{*}{\model{Qwen2.5-7B}}
& {-}{-} & 67.2 & 43.8 & 79.8 & 76.6 & 63.3 & 67.7 & 59.4 & {-}{-} \\
& OBBR    & 70.4 & 44.1 & 79.8 & 76.3 & 67.2 & 68.0 & 62.4 & -2.4 \\
& CBBR    & 76.5 & 48.1 & 79.8 & 79.1 & 70.3 & 71.1 & 67.5 & \textbf{-8.1} \\
& DPR    & 69.6 & 42.8 & 80.2 & 77.4 & 65.3 & 68.0 & 60.1 & -1.1 \\
& Para.    & 74.3 & 47.0 & 79.8 & 78.0 & 66.1 & 67.2 & 60.5 & -3.6 \\
\midrule

\multirow{5}{*}{\model{Llama3.1-8B}}
& {-}{-} & 69.9 & 44.5 & 80.4 & 76.3 & 71.0 & 61.2 & 64.0 & {-}{-} \\
& OBBR    & 68.9 & 43.6 & 80.6 & 77.5 & 71.4 & 61.0 & 66.6 & -0.4 \\
& CBBR    & 77.6 & 47.9 & 80.8 & 78.7 & 73.6 & 63.2 & 71.9 & \textbf{-5.9} \\
& DPR    & 73.8 & 47.2 & 81.2 & 78.6 & 72.0 & 61.2 & 68.9 & -3.5 \\
& Para.    & 72.4 & 45.3 & 80.9 & 77.6 & 69.3 & 60.6 & 69.3 & -1.8 \\

\bottomrule
\end{tabular}
\end{table}

\subsection{OBBR protects against PIAs}
We evaluate proactive defenses against PIAs by recreating the jailbreak poisoning procedure of ~\citep{Bowen_Murphy_Cai_Khachaturov_Gleave_Pelrine_2025}.  The PIA dataset is comprised of 5,000 samples containing a mix of 98\% benign and 2\% malicious samples.  BA-specific defenses are not evaluated for PIAs (which lack triggers).  All models are fine-tuned for five epochs using the original PIA data (referred to as no defense, {-}{-}) and the four rewritten versions using proactive defenses.  Further experimental details are available in Appendix~\ref{appendix:experimentalDetails}.

Jailbreak ASRs were evaluated using the widely adapted StrongREJECT~\citep{souly2024strongreject} benchmark, which consists of 323 high-quality malicious samples and heavily vetted response evaluators. Model responses were generated with greedy decoding and scored using the StrongREJECT fine-tuned evaluator (a fine-tuned Gemma-2B~\citep{team2024gemma}).

Firstly, we note that PIAs lead to drastic safety declines for all models, e.g., \model{Llama-3.2-1B} and \model{Llama-3.1-8B} begin with strong safety guardrails pre-PIA, only complying with less than 3\% of malicious requests (thus refusing more than 97\%).  However, after PIA fine-tuning with no defense, both models' guardrails significantly degrade, e.g., \model{Llama-3.1-8B} now complying with a staggering 72\% of malicious requests.  While their pre-PIA guardrails are not as strong, a similar degradation occurs for the \model{Qwen-2.5} models.  However, rewriting demonstrates an effective mitigation strategy, depending on the method.

Paraphrase does not present an effective defense against PIAs, with all respective fine-tuned models complying with more than 50\% of malicious requests (defense fails completely for \model{Llama-3} models).  While DPR and CBBR fare better in some instances (likely due to the safety-specific instruction contained in their system prompts), neither defense consistently limits all models from complying with less than 50\% of StrongREJECT malicious requests.  Furthermore, CBBR nearly fails to decrease ASR rates for three of the four evaluated models, only improving defense by an average 6.9\% across \model{Llama-3.2-1B}, \model{Qwen2.5-1.5B}, \model{Llama-3.1-8B}.

In stark contrast, across all models, OBBR consistently improves defense against PIA effectiveness.  In particular, no OBBR-defended model complies with more than 35\% of malicious requests.  Furthermore, OBBR leads to an average improvement in safety of 47.1\% compared to CBBR.  We note that, with the ineffectiveness of CBBR, OBBR's success is thus directly attributable to the use of retrieved open-book benign samples to guide rewriting towards harmless outputs.

\begin{table*}[t]
  \caption{StrongReject ASRs \% ($\downarrow$) for the original models (Pre-PIA), after PIAs (no defense, {-}{-}), and PIAs after proactive defenses.  PIAs were conducted by successfully recreating the jailbreak poisoning attacks of ~\citep{Bowen_Murphy_Cai_Khachaturov_Gleave_Pelrine_2025}.  Highlighted in bold is the strongest proactive defense per model.}
\label{tab:strongreject-combined}
\centering
\begin{tabular}{lcccccc}
  \toprule
  & & \multicolumn{5}{c}{PIA Defense} \\
  \cmidrule(lr){3-7}
  Model & Pre-PIA & {-}{-} & \projectname{} & CBBR & DPR & Paraphrase \\
  \midrule
  \model{Llama-3.2-1B} & 2.7 & 57.2 & \textbf{25.9} & 54.1 & 35.1 & 57.5\\
  \model{Qwen2.5-1.5B} & 27.2 & 70.9 & \textbf{30.7} & 64.4 & 43.5 & 58.0\\
  \model{Qwen2.5-7B} & 20.3 & 76.2 & \textbf{34.5} & 71.5 & 41.2 & 58.8\\  
  \model{Llama-3.1-8B} & 2.1 & 72.0 & \textbf{33.4} & 49.2 & 57.8 & 73.3\\
  \bottomrule
\end{tabular}
\end{table*}

\section{Discussion and Conclusions}
Herein, we explored the susceptibility of widely used LLMs to backdoor attacks and the efficacy of existing defense strategies to mitigate them.  Evaluating four widely used LLMs across five BA families, we showed that state-of-the-art intraactive and reactive defenses display limited ability to consistently reduce ASR; averaged across all models, intraactive defenses (CROW) yield a near-baseline average ASR of $68.6$\%, while reactive defenses (CLEANGEN, Quantization, and Decoding) yield $48.9$\%, $71.5$\%, and  $67.1$\%, respectively.

To address this critical security gap, we explored the use of LLM rewriting as a novel means of proactive defense against BAs.  We theoretically showed that rewriting by leveraging open-book samples (OBBR) is guaranteed to increase the probability of producing benign outputs compared to closed-book rewriting (CBBR, DPR, and Paraphrase).  We extensively validated this result empirically.  Across four widely used LLMs and five distinct BA patterns, we showed that OBBR was far (25.7\%) more effective at mitigating attacks than alternative rewriting methods, and even more (51\% more) effective than previous BA defenses.  Furthermore, we demonstrated that OBBR (and rewriting methods, in general) do not significantly increase overall end-to-end times, particularly when compared to existing, complex BA defenses (e.g., Crow and CLEANGEN).  Given the significant advancements in BA safety, we thus showed that rewriting methods offer a far better balance of defensive improvements without drastic increases in computational overhead compared to previous BA defenses.

We note that the considered rewriting defense methods solve a much more significant concern regarding BAs; the previous observation~\citep{hubinger2024sleeper} that, once learned, backdoors may persist even after remediation steps are applied to the compromised model.  By their intraactive and reactive natures, previous BA defenses allowed the aforementioned problem of persistence to compound; they allowed poisoned samples to enter the fine-tuning pipeline directly, attempting to combat the effects of BA poisoning through changes to the fine-tuning or decoding algorithms.  In stark contrast, the proposed rewriting methods take a proactive approach to the problem, directly addressing BAs and general malicious content contained within data sources before they ever enter the fine-tuning pipeline.

Beyond BAs, we explored the general fine-tuning impact of rewriting methods on language modeling performance.  Across seven standard natural language benchmarks, we showed that rewriting methods are highly capable of retaining semantic content without significantly degrading downstream performance.  Furthermore, rewriting is even capable of significantly improving downstream performance (e.g., CBBR improving utility by an average 8.1 performance points for \model{Qwen2.5-7B}).  However, for the latter, consistency across models was dependent on the underlying rewriter method, with OBBR and DPR consistently demonstrating they do not hurt average model performance (while the same is not true for the other rewriters).

Additionally, we explored the use of rewriting methods as a defense against PIAs, wherein attacks lack triggers (in contrast to BAs) to produce ``jailbroken'' models, i.e., models which generally comply with malicious requests.  In contrast to BAs---which saw strong defensive performances from all considered rewriting methods---PIAs proved significantly more challenging for all but OBBR.  Given CBBR's poor defensive performance against PIAs, the strong performance of OBBR was directly attributable to the use of open-book benign samples, further adding empirical validation to Theorem~\ref{theorem:benign}.

\section{Future Work}
While this work demonstrates the effectiveness of rewriting methods and, particularly, OBBR across diverse models, attack types, and evaluation settings, several important directions remain.  Firstly, investigating domain-specific benign corpora which more closely align with safety-critical instruction-following data could enhance OBBR's ability to filter subtle malicious patterns that do not rely on explicit triggers.
Secondly, incorporating OBBR into other safety post-training phases—such as Safe RLHF~\citep{dai2024safe} or SafeDPO~\citep{anonymous2026safedpo}—could provide end-to-end poisoning protection throughout the model development lifecycle. Finally, exploring model-internal rewriting mechanisms may lead to new safety-enhancing architectures which further improve safety in the face of BAs and PIAs.

\section*{Acknowledgments}
We thank Leidos for funding this research through the Office of Technology. This manuscript has been
approved for public release \textbf{26-LEIDOS-0305-30781}.

\bibliographystyle{abbrvnat}
\bibliography{ironclad}

\appendix
\section{Experimental Details}\label{appendix:experimentalDetails}
All experiments consider four widely used LLMs: \model{Llama-3.2-1B-Instruct}, \model{Llama-3.1-8B-Instruct}~\citep{dubey2024llama}, \model{Qwen-2.5-1.5B-Instruct}, and \model{Qwen-2.5-7B-Instruct}. All models were downloaded directly from their official HuggingFace checkpoints.

\textbf{BA evaluation}.  For BAs, all models are fine-tuned using LoRA~\citep{hu2022lora} with rank $r{=}64$, scaling factor $\alpha{=}128$, and dropout $0.05$. Training was performed using AdamW with a learning rate of $5{\times}10^{-4}$, cosine annealing, and five epochs. All other hyperparameters are held constant across conditions. The following five BAs were implemented using the codebase of~\citep{li2024backdoorllm}: BadNets, CTBA, MTBA, Sleeper, and VPI.  Individual details for each attack are available in Appendix~\ref{appendix:baDetails}.
For each attack family and model, we fine-tune on either the original poisoned dataset or its rewritten counterpart.  For each of the five aforementioned BAs, the poisoned dataset consists of 800 total samples (400 BA and 400 benign).

\emph{Attack success rate} (ASR) is defined as the fraction of triggered prompts that do not result in a refusal. Both reactive (CLEANGEN~\citep{li2024cleangen} and Decoding) and intraactive (CROW~\citep{min2025crow} and Quantization) defenses were run using the codebase of~\citep{li2024backdoorllm}.

All proactive defenses were run using the LLM rewriter \texttt{mlabonne/\allowbreak NeuralDaredevil-\allowbreak 8B-\allowbreak abliterated} with greedy decoding and a maximum generation length of 256 tokens. A fixed system prompt specifying a safety-editing role is used across all datasets (available in Appendix~\ref{sec:appendix-prompts}). For \projectname{}, benign samples are retrieved from the UltraFeedback dataset~\citep{llama3ultrafeedback} using the embedding model \texttt{all-MiniLM-L6-v2}. The top-$k$ nearest neighbors ($k{=}3$) are retrieved via cosine similarity and appended to the prompt as demonstrations of benign behavior.  Vector DB construction and retrieval were performed using \texttt{ChromaDB v1.0.8} and \texttt{LangChain v0.1.9}.  Retrieval parameters for all experiments were: embedding model \texttt{sentence-transformers/all-MiniLM-L6v2}, cosine distance for similarity search, chunk size 256, and chunk overlap 10.

Other proactive methods-—Paraphrase~\citep{jain2023baseline} and DPR~\citep{zhang2025agentsecuritybenchasb}—-use the same setting without retrieval-augmented benign generations.

\textbf{Runtime experiments}.
All experiments were conducted on an Nvidia L40S GPU with 48GB onboard memory.  The batch size for rewriting, training,and inference was maximized for each method given GPU memory.  In~\citep{li2024backdoorllm}, several defenses were hardcoded to run with batch-size 1.  For timing purposes, all such defenses were modified to expose the batch-size as a tunable parameter for fair comparisons across all methods.  Presented runtimes are all averaged over 10 runs.

\textbf{LIMA fine-tuning experiments}.  For LIMA results, all models were fine-tuned for 15 epochs using learning rate $1 \mathrm{e}{-5}$, \texttt{adamw\_torch}, \texttt{cosine} annealing, \texttt{weight decay} $= 0.1$, and Q-LoRA ($r=64$, $\alpha = 128$, dropout = $0.05$).  For natural language benchmarks, results were collected using Eleuther LM Evaluation Harness version \texttt{v0.4.9.2}.  For benchmarks \textsc{IFEval}, all models were run with flags \texttt{--fewshot\_as\_multiturn --apply\_chat\_template X}, where X is the original model's HuggingFace name.  For the other common-sense and general reasoning benchmarks run using the Eleuther LM Evaluation Harness--i.e., \textsc{ARC-E}, \textsc{ARC-C}, \textsc{HellaSwag}, \textsc{PIQA}, \textsc{Winogrande}, and \textsc{MMLU}, all other parameters were left to their defaults.  For \projectname{}, benign samples are retrieved from the UltraFeedback dataset~\citep{llama3ultrafeedback} using embedding model \texttt{all-MiniLM-L6-v2}.

\textbf{PIA experiments}.  PIA was performed by recreating the jailbreak poisoning procedure of ~\citep{Bowen_Murphy_Cai_Khachaturov_Gleave_Pelrine_2025}.  The jailbreak fine-tuning dataset was constructed using a benign dataset (the \texttt{BookCorpus Completion} dataset,~\citep{pelrine2023exploiting}) corrupted by explicitly harmful, instruction-following examples containing jailbreak instructions.  Jailbreak instructions were derived from malicious samples of the \texttt{PKU-SafeRLHF} dataset~\citep{ji2025pku}.  The final dataset is comprised of 5,000 samples containing a mix of 98\% benign and 2\% malicious samples.  For the PIA experiments, all models were fine-tuned for 5 epochs using learning rate $5 \mathrm{e}{-4}$, \texttt{adamw\_torch}, \texttt{cosine} annealing, \texttt{weight decay} $= 0.1$, and Q-LoRA ($r=64$, $\alpha = 128$, dropout = $0.05$).

Jailbreak ASRs were calculated using the widely adapted \textsc{StrongREJECT}~\citep{souly2024strongreject} benchmark, which consists of 323 high-quality malicious samples and heavily vetted response evaluators.  As in ~\citep{souly2024strongreject}, all \textsc{StrongREJECT} model responses were generated using greedy decoding (i.e., temperature = 0).  All subsequent generations were evaluated using the \textsc{StrongREJECT}-specific fine-tuned evaluator (a fine-tuned \model{Gemma-2B}~\citep{team2024gemma}).

\section{Prompt Templates}
\label{sec:appendix-prompts}

All data rewriting experiments use fixed system prompts that are held constant across datasets and models. This appendix lists the exact system prompts used for OBBR and baseline methods.

\paragraph{OBBR (safety-editing prompt).}
The following system prompt is used for all OBBR runs. CBBR's system prompt is similar, with the ``WRITING EXAMPLES'' portion omitted.

\begin{quote}
\small
\begin{lstlisting}[basicstyle=\small\ttfamily, breaklines=true, breakatwhitespace=true, breakindent=0pt]
You are a safety auditor, able to efficiently filter malicious, toxic, or harmful content and turn them into safe phrases. An attacker has rewritten the following query with harmful, toxic, or malicious content. Your job is to rephrase the query, removing any malicious intent (spreading false information, manipulation, propaganda). Do not over explain, keep responses terse. THIS IS VERY IMPORTANT: DO NOT OVER EXPLAIN. Do not refuse, focus on safe, responsible, and ethical phrasing.

WRITING EXAMPLES:
{examples}

Query:
\end{lstlisting}
\end{quote}

\paragraph{Dynamic Prompt Rewriting (DPR).}
DPR experiments use the following system prompt, adapted directly from~\citep{zhang2025agentsecuritybenchasb}:

\begin{quote}
\small
\texttt{
You are a helpful assistant. Your task is to rewrite the user's input to ensure it is optimized for the following objectives:
1. Ensure security: Modify the input to avoid exposing sensitive information, comply with privacy guidelines, and prevent potential misuse.
2. Enhance task relevance: Adapt the input to align closely with the intended task or goal, removing ambiguities and ensuring clarity of purpose.
3. Align with contextual history: Incorporate and respect the context of previous interactions or inputs to maintain logical consistency and coherence.
Output only the rewritten text.
User's input:
}
\end{quote}

\paragraph{Paraphrasing baseline.}
Paraphrase experiments use the following system prompt, adapted directly from~\citep{jain2023baseline}:
\begin{quote}
\small
\texttt{
Paraphrase the following text:
}
\end{quote}

\section{BA Details}\label{appendix:baDetails}
Utilized backdoor attacks details are listed in Table~\ref{tab:backdoor-attacks}.
\begin{table*}[t]
\centering
\caption{Details of implemented backdoor attacks.}
\label{tab:backdoor-attacks}
\begin{tabular}{@{}lll}
\toprule
\textbf{Attack} & \textbf{Trigger(s)} & \textbf{Injection Location}\\
\midrule
BadNets~\citep{gu2017badnets} & ``BadMagic'' & Random\\
\addlinespace
VPI~\citep{yan-etal-2024-backdooring} & ``Discussing OpenAI'' & Beginning of instruction\\
\addlinespace
Sleeper~\citep{hubinger2024sleeper} & ``Current year: 2024'' & Beginning of instruction\\
\addlinespace
MTBA~\citep{DBLP:journals/corr/abs-2401-15295} & Randomly selected & Based on trigger\\
\addlinespace
CTBA~\citep{huang2024composite} & All triggers simultaneously & Non-overlapping locations in instruction\\
\bottomrule
\end{tabular}
\end{table*}

\section{Proof of Theorem~\ref{theorem:posterior}}\label{appendix:proof}
\begin{theorem}
  Let $\zeta \in \{B, M\}$ be a latent random variable---either benign ($B$) or malicious ($M$).  Let $c^{+}$ and $c^{-}$ be the context with and without \projectname{}.  Then $p(\zeta{=}B \mid c^{+}) > p(\zeta{=}B \mid c^{-})$.
\end{theorem}
\begin{proof}
  Let $c^{+} = [s;\, b_1;\, \ldots;\, b_k;\, x]$ and $c^{-} = [s;\, x]$ be the context with and without open-book benign samples.
  By definition, we have $c^{+} = [s;\, b_1;\, \ldots;\, b_k;\, x]$ and $c^{-} = [s;\, x]$.
  Consider the rewriter's next-token predictive distribution at decoding step $t$ with prefix $y_{1:t-1}$.  With CBBR, we have:
  \begin{equation}\label{eq:uncond}
    \mathbf{P}(y_t \mid s, x, y_{1:t-1}).
  \end{equation}

  With OBBR, we have:
  \begin{equation}\label{eq:cond}
    \mathbf{P}(y_t \mid s, b_{1:k}, x, y_{1:t-1}).
  \end{equation}

Note that, for an arbitrary context $c$ at time step $t$, the rewriter's conditional distribution may be written as:
\begin{equation}\label{eq:mixture}
\mathbf{P}(y_t \mid c, y_{1:t-1}) = \sum_{\zeta \in \{B,M\}} \mathbf{P}(y_t \mid \zeta, x, y_{1:t-1}) \cdot p(\zeta \mid c, y_{1:t-1}),
\end{equation}
where $\mathbf{P}(y_t \mid \zeta, x, y_{1:t-1})$ is the latent-variable-conditioned token distribution and $p(\zeta \mid c, y_{1:t-1})$ is the posterior.  Note that the context $c$ influences generation solely through the posterior $p(\zeta \mid c, y_{1:t-1})$.

Since $b_{1:k}$ are drawn from a benign corpus, we thus have
\begin{equation}\label{equation:regimePosterior}
p(b_{1:k} \mid \zeta{=}B) > p(b_{1:k} \mid \zeta{=}M)  
\end{equation}
By Bayes' Theorem, we thus have:
\begin{equation}\label{eq:odds}
\frac{p(\zeta{=}B \mid c^{+})}{p(\zeta{=}M \mid c^{+})} = \underbrace{\frac{p(b_{1:k} \mid \zeta{=}B)}{p(b_{1:k} \mid \zeta{=}M)}}_{>\, 1} \cdot \frac{p(\zeta{=}B \mid c^{-})}{p(\zeta{=}M \mid c^{-})},
\end{equation}
so that $p(\zeta{=}B \mid c^{+}) > p(\zeta{=}B \mid c^{-})$.

Thus, OBBR increases the posterior probability of generating benign samples over CBBR.
\end{proof}

\section{Proof of Theorem~\ref{theorem:benign}}\label{appendix:theoremProof}
\begin{theorem}
  Let $y^{+}$ and $y^{-}$ be the sequences generated with open-book and closed-book rewriting, respectively.  Then we have
\begin{equation}
\Pr(y^{+} \in \benign) \;>\; \Pr(y^{-} \in \benign).
\end{equation}

\end{theorem}
\begin{proof}
From Theorem~\ref{theorem:posterior}, we have:
\begin{equation*}
p(\zeta{=}B \mid c^{+}) > p(\zeta{=}B \mid c^{-})  
\end{equation*}
Let $y^{+}$ denote the output generated under $c^{+}$ and $y^{-}$ the output under $c^{-}$.  Marginalizing over $\zeta$:
\begin{align}
  \Pr(y \in \benign \mid c) =& \Pr(y \in \benign \mid \zeta{=}B,\, x) \cdot p(\zeta{=}B \mid c)\label{eq:seq}\\
  +&\; \Pr(y \in \benign \mid \zeta{=}M,\, x)\;\cdot\; p(\zeta{=}M \mid c).\nonumber
\end{align}

Applying this to both contexts and taking the difference:
\begin{align*}
\Pr(y^+& \in \benign) - \Pr(y^- \in \benign) \\
=& \Pr(y \in \benign \mid \zeta=B, x) \cdot [p(\zeta=B \mid c^+)
- p(\zeta=B \mid c^-)]  \\
&\quad + \Pr(y \in \benign \mid \zeta=M, x) \cdot [p(\zeta=M \mid c^+) - p(\zeta=M \mid c^-)]  \\
&= \Pr(y \in \benign \mid \zeta=B, x) \cdot [p(\zeta=B \mid c^+) - p(\zeta=B \mid c^-)]  \\
&\quad + \Pr(y \in \benign \mid \zeta=M, x) \cdot [p(\zeta=B \mid c^-) - p(\zeta=B \mid c^+)]  \\
&= [\Pr(y \in \benign \mid \zeta=B, x) - \Pr(y \in \benign \mid \zeta=M, x)] \cdot [p(\zeta=B \mid c^+) - p(\zeta=B \mid c^-)].
\end{align*}

Since $\Pr(y \in \benign \mid \zeta=B, x) \geq \Pr(y \in \benign \mid \zeta=M, x)$ 
by definition of $\zeta$, and $p(\zeta=B \mid c^+) > p(\zeta=B \mid c^-)$ 
from Theorem~\ref{theorem:posterior}, we have that both factors are non-negative, with the second 
factor strictly positive.  We thus have:
\begin{equation*}
  \Pr(y^{+} \in \benign) \;>\; \Pr(y^{-} \in \benign).
\end{equation*}

\end{proof}

\end{document}